\pdfoutput=1
\RequirePackage{ifpdf}
\ifpdf 
\documentclass[pdftex]{sigma}
\else
\documentclass{sigma}
\fi

\numberwithin{equation}{section}

\newtheorem{Theorem}{Theorem}[section]
\newtheorem{Corollary}[Theorem]{Corollary}
\newtheorem{Lemma}[Theorem]{Lemma}
\newtheorem{Proposition}[Theorem]{Proposition}
 { \theoremstyle{definition}
\newtheorem{Remark}[Theorem]{Remark} }

\usepackage{bm}

\renewcommand{\d}{\mathrm d}
\newcommand{\pa}{\partial}

\def \pa{\partial}
\def\C{{\mathbb C}}

\def \YY{\mathbf Y}

\def \gg{\mathbf g}

\def\PP{\bm{P}}
\def\hPP{\bm{\hat{P}}}
\def\WW{\bm{W}}
\def\FF{\bm{F}}

\def\Id{\bm{I}}
\def\YY{\bm{Y}}
\def\TT{\bm{T}}
\def\BB{\bm{B}}
\def\QQ{\bm{Q}}
\def\HH{\bm{H}}
\def\LL{\bm{L}}
\def\gg{\bm{\gamma}}
\def\kk{\bm{\kappa}}
\def\cC{\mathcal{C}}
\def\Ppsi{\bm{\Psi}}
\def\RR{\bm{R}}
\def\AA{\bm{A}}
\def\cAA{\bm{\mathcal{A}}}
\def\BB{\bm{B}}
\def\cBB{\bm{\mathcal{B}}}
\def\UU{\bm{\mathcal{U}}}
\def\SS{\bm{S}}
\def\aa{\bm{a}}
\def\bb{\bm{b}}
\def\pp{\bm{p}}
\def\qq{\bm{q}}
\def\aalpha{\bm{\alpha}}
\def\bbeta{\bm{\beta}}
\def\ssigma{\bm{\sigma}}

\def\CC{\bm{\mathcal C}}
\def\DD{\bm{\mathcal D}}

\begin{document}

\allowdisplaybreaks

\newcommand{\arXivNumber}{1801.08740}

\renewcommand{\thefootnote}{}

\renewcommand{\PaperNumber}{076}

\FirstPageHeading

\ShortArticleName{Toda and Painlev\'e Systems Associated with Semiclassical MVOPs of Laguerre Type}

\ArticleName{The Toda and Painlev\'e Systems Associated\\ with Semiclassical Matrix-Valued Orthogonal\\ Polynomials of Laguerre Type\footnote{This paper is a~contribution to the Special Issue on Painlev\'e Equations and Applications in Memory of Andrei Kapaev. The full collection is available at \href{https://www.emis.de/journals/SIGMA/Kapaev.html}{https://www.emis.de/journals/SIGMA/Kapaev.html}}}

\Author{Mattia CAFASSO~$^\dag$ and Manuel D.~DE LA IGLESIA~$^\ddag$}

\AuthorNameForHeading{M.~Cafasso and M.D.~de la Iglesia}

\Address{$^\dag$~LAREMA - Universit\'e d'Angers, 2 Boulevard Lavoisier, 49045 Angers, France}
\EmailD{\href{mailto:cafasso@math.univ-angers.fr}{cafasso@math.univ-angers.fr}}
\URLaddressD{\url{https://sites.google.com/site/mattiacafasso/}}

\Address{$^\ddag$~Instituto de Matem\'aticas, Universidad Nacional Aut\'onoma de M\'exico,\\
\hphantom{$^\ddag$}~Circuito Exterior, C.U., 04510, Mexico City, Mexico}
\EmailD{\href{mailto:mdi29@im.unam.mx}{mdi29@im.unam.mx}}
\URLaddressD{\url{http://www.matem.unam.mx/mdi29/}}

\ArticleDates{Received March 28, 2018, in final form July 16, 2018; Published online July 21, 2018}

\Abstract{Consider the Laguerre polynomials and deform them by the introduction in the measure of an exponential singularity at zero. In [Chen Y., Its A., \textit{J.~Approx. Theo\-ry} \textbf{162} (2010), 270--297] the authors proved that this deformation can be described by systems of differential/difference equations for the corresponding recursion coefficients and that these equations, ultimately, are equivalent to the Painlev\'e~III equation and its B\"acklund/Schlesinger transformations. Here we prove that an analogue result holds for some kind of semiclassical matrix-valued orthogonal polynomials of Laguerre type.}


\Classification{34M56; 35Q15; 37J35; 42C05}

\renewcommand{\thefootnote}{\arabic{footnote}}
\setcounter{footnote}{0}


\section{Introduction}

The relation between integrable systems and orthogonal polynomials goes back to the seminar paper of Moser \cite{Moser}. There, the author made the remarkable observation that, given a family of orthogonal polynomials associated to a measure $w$ on the real axis, the related tridiagonal Jacobi matrix can be interpreted as a Lax matrix for the Toda equations, where the independent variable of equations plays the role of deformation parameters for the measure $w$. Starting from the nineties, this result and various generalizations (to orthogonal polynomials on the unit circle, to multiple orthogonal polynomials and many other types of special polynomials) played a central role in the field of random matrices, and the study of the equations satisfied by orthogonal polynomials found applications to 2D quantum gravity \cite{FIK} and to the computations of gap probabilities and partition functions associated to random models, both in the discrete and continuous setting (see, for instance, \cite{AvMV} and \cite{Joh}).

The scope of this paper is to work on the relation between \emph{matrix-valued} orthogonal polynomials and \emph{non-commutative} integrable equations. This is a relatively new field, in which few examples have been worked out. In~\cite{M1}, it had been shown that matrix-valued orthogonal polynomials (on the real line) satisfy a non-commutative version of the Toda equations, and in \cite{Caf2} one of the authors provided an analogue result for matrix-valued orthogonal polynomials on the unit circle, upon replacing Toda equations with the (multi-component) Ablowitz--Ladik hierarchy (a~well known integrable discretisation of the non-linear Schr\"odinger one). In both cases, the main tool had been the theory of quasi-determinants (Schur complements), which were known to be related to non-commutative integrable systems since the pioneering work of Etingof, Gelfand and Retakh~\cite{EGR}. Schur complements had been used also in the series of papers~\cite{AFM,AM1,AM2} and in~\cite{RR}, where a non-commutative version of the Painlev\'e II equation was introduced, and studied later on in \cite{BertolaCafasso} using Riemann--Hilbert techniques.

In 2011, independently in \cite{CassaM} and \cite{OPRH}, it was established a general theory to study matrix-valued orthogonal polynomials through Riemann--Hilbert problems, thus extending the well known result of Fokas, Its and Kitaev~\cite{FIK} from the scalar to the matrix case. This theory, in our opinion, provides one of the best ways to deduce, in a uniform way, non-linear equations related to matrix-valued orthogonal polynomials, using the well known technique of reducing the Riemann--Hilbert problem to a simpler one, where the jumps are constant. In~\cite{CassaM}, this technique has been used to deduce a matrix version of the discrete Painlev\'e~I equation, while in~\cite{CafassodelaIglesia} we proved that the Christoffel--Darboux kernel associated to certain matrix-valued Hermite polynomials satisfies a non-commutative version of the Painlev\'e~IV equation.

The goal of this paper is to give another instance of this relation between non-commutative Painlev\'e equations and matrix-valued orthogonal polynomials. The results we present are to be thought as a non-commutative analogue of the results in~\cite{ChenIts}. There, the authors studied the (scalar) orthogonal polynomials associated to the measure
$w(x,s) := x^\alpha {\rm e}^{-x - s/x}$ on~$\mathbb R_+$, which is a~deformation of the usual Laguerre measure, where the deformation is induced by the parameter~$s$ and changes the behaviour of the measure at zero. Using two different approaches (ladder operators and Riemann--Hilbert techniques) Chen and Its proved that some quantities~$a_n(s)$ and~$b_n(s)$, which ultimately can be expressed through the recursion coefficients of the orthogonal polynomials, satisfy a differential and a difference system, which can be reduced, respectively, to the Painlev\'e~III equation and a discrete analogue of it, which is conjectured to be a composition of the basic Schlesinger transformations of Painlev\'e~III.

The non-commutative analogue of the difference and differential systems for $a_n$ and $b_n$ are given here in Theorems \ref{discreteThm} and \ref{continuousThm}. The derivation, which is based upon the Riemann--Hilbert method established in \cite{OPRH}, is not a straightforward generalisation of the one for the scalar case: namely one has to push a little bit further the analysis of the behavior of the solution of the Riemann--Hilbert problem at the two singular points~$0$ and $\infty$. Also, while the systems of $a_n$ and $b_n$ were of first-order, here our two systems are of second-order, which is a phenomenon that, in some way, we already observed in the previous paper~\cite{CafassodelaIglesia} (and also, to some extent, in~\cite{BertolaCafasso} for the case of Painlev\'e XXXIV).

The paper is organized as follows: in the first section we define the matrix-valued orthogonal polynomials we want to study, construct the related Riemann--Hilbert problem and the corresponding Lax triple~\eqref{Laxtriple}. In the second section, the compatibility conditions of the Lax triples are computed, together with some additional relations that, in the scalar case, are deduced from the compatibility conditions but here they have to be computed in a different way (see Proposition~\ref{propg}). Finally, in the third section, our main results, Theorems~\ref{discreteThm} and~\ref{continuousThm}, are stated and deduced, with straightforward (but sometimes lengthy) computations. The last section gives some additional relations holding for a special class of matrix-valued orthogonal polynomials that were introduced in~\cite{DG1}.

\section{From the Riemann--Hilbert problem to the Lax system}

Consider the following $N\times N$ weight matrix
\begin{gather*}
 \WW(s;x) := x^\alpha{\rm e}^{-x - s/x}\TT(x)\TT^*(x),\qquad x\in[0,\infty),\qquad \alpha>0,\qquad s>0.
\end{gather*}
Observe that $\WW(s;x)$ is a weight matrix of Laguerre type perturbed by a multiplicative factor~${\rm e}^{-s/x}$, which induces an infinitely strong zero at the origin. For simplicity, the matrix-valued function $\TT(x)$ is chosen such that
\begin{gather}\label{Tcondition}
 \big(\partial_x \TT(x)\big)\TT^{-1}(x) = \frac{\BB}x,
\end{gather}
where $\BB$ is an arbitrary $N\times N$ constant matrix, independent of $s$, such that all moments of $\WW(s;x)$ are finite (entrywise). Observe that in this case, the solution of~\eqref{Tcondition} is given by $\TT(x)=x^{\BB}={\rm e}^{\BB\log x}$.

Given such a weight we construct (if existing) a sequence of \emph{monic} matrix-valued orthogonal polynomials $\{\hPP_n(s;x)\}_{n \geq 0}$ which is uniquely characterized by these two conditions:
\begin{itemize}\itemsep=0pt
\item $\hPP_n(s;x)$ is monic of order $n$:
\begin{gather*}
\hPP_n(s;x) = x^n \Id_N + \cdots,
\end{gather*}
\item For any $n,m \geq 0$,
\begin{gather*}
\int_{0}^\infty \hPP_n(s;x) \WW(s;x) \hPP^*_m(s;x) \d x = \gg_n^{-1}(s)\delta_{n,m},
\end{gather*}
where $\gg_n(s)$ is the inverse of the $n^{th}$ (matrix-valued) norm of our family of matrix-valued orthogonal polynomials. Note that, for any $n$, $\gg_n(s)$ is a Hermitian matrix.
\end{itemize}

As in the scalar case, the existence of the sequence of matrix-valued orthogonal polyno\-mials can be equivalently restated as the existence of the solution to a certain (matrix-valued) Riemann--Hilbert boundary value problem, as it has been proven in \cite{OPRH}. This is a general result that does not depend on the particular form of the measure $\WW(s;x)$ and we make reference to~\cite{OPRH} for its general formulation. Applied to our case, we have the following proposition:
\begin{Proposition}\label{propLax} The sequence of matrix-valued orthogonal polynomials $\{ \hPP_n(s;x)\}_{n \geq 0}$ exists if and only if, for any $n \geq 0$, it exists a block matrix-valued function $\YY^{(n)}(s;z)$ which is analytic on $\C\setminus{[0,\infty)}$ and such that the following three conditions are satisfied:
 \begin{itemize}\itemsep=0pt
 \item The boundary values $\YY_+^{(n)}$, $\YY_-^{(n)}$ on the $($standardly oriented$)$ contour $(0,\infty)$ satisfy the following jump condition:
 \begin{gather}\label{jumpcondition}
 \YY_+^{(n)}(s;x) = \YY_-^{(n)}(s;x) \left( \begin{matrix} \Id_N & x^\alpha{\rm e}^{-x - s/x}\TT(x)\TT^*(x) \\
 0 & \Id_N \end{matrix} \right), \qquad x \in (0,\infty).
 \end{gather}
 \item At infinity the solution $\YY^{(n)}$ has the following asymptotic behaviour:
 \begin{gather}\label{asympinfty}
 \YY^{(n)}(s;z) = \left( \Id_{2N} + \sum_{k = 1}^\infty \frac{\YY_{-k}^{(n)}(s)}{z^{k}} \right) \left( \begin{matrix} z^n\Id_N & 0\\
 0 & z^{-n}\Id_N \end{matrix}\right), \qquad z\rightarrow \infty.
 \end{gather}
 \item At the point $0$ the solution $\YY^{(n)}$ is non--singular:
 \begin{gather}\label{asymp0}
 \YY^{(n)}(s;z) = \QQ^{(n)}(s) \left( \Id_{2N} + \sum_{k = 1}^\infty \YY^{(n)}_k(s) z^k \right), \qquad z\rightarrow 0.
 \end{gather}
 \end{itemize}
\end{Proposition}
Indeed, one can write the solution of the Riemann--Hilbert problem above in terms of matrix-valued orthogonal polynomials\footnote{While $\hPP_n(s;x)$ is, strictly speaking, defined on the positive real axis, below we denote with $\hPP_n(s;z)$ its analytic continuation on the complex plane. Also, we adopt the convention that for any matrix-valued function~$\bm P(z)$, we have that $\bm P^*(z):=\left(\bm P(\bar z)\right)^*$.} as follows
\begin{gather}\label{solutionRHP}
 \YY^{(n)}(s;z) = \left( \begin{matrix} \hPP_n(s;z) & \cC(\hPP_n\WW)(s;z)\\
 -2\pi i \gg_{n-1}(s)\hPP_{n-1}(s;z) & -2\pi i \gg_{n-1}(s) \cC(\hPP_{n-1}\WW)(s;z)
 \end{matrix}\right).
\end{gather}
Here we denoted with $\cC$ the Cauchy transform, applied to (possibly) matrix-valued func\-tions $\FF(x)$ in such a way that
\begin{gather*}
 \cC(\FF)(z) := \frac{1}{2 \pi i} \int_{0}^\infty \frac{\FF(x)}{x - z}\d x.
\end{gather*}
We can normalize our polynomials in the following way
\begin{gather*}
 \PP_n(s;x) = \kk_n(s)\hPP_n(s;x), \qquad \mbox{with} \quad \gg_{n}(s) = \kk^*_n(s)\kk_n(s).
\end{gather*}

As it is customary, in order to derive differential and difference equations from our family of matrix-valued orthogonal polynomials we will reduce the Riemann--Hilbert problem \eqref{jumpcondition}--\eqref{asymp0} to a one with \emph{constant} jumps for $\Ppsi^{(n)}(s;z) := \YY^{(n)}(s;z) \RR(s;z)$, with $\RR(s;z)$ to be defined. Then $\Ppsi^{(n)}(s;z)$ will satisfy a Lax system of three equations, whose coefficients will depend on the entries of $\YY^{(n)}_{-1}(s)$ and $\QQ^{(n)}(s)$ defined below. The two propositions and the corollary below express $\YY^{(n)}_{-1}(s)$ and $\QQ^{(n)}(s)$ in function of meaningful quantities for the corresponding matrix-valued orthogonal polynomials. In the following, we denote
\begin{gather*}
 \hPP_n(s;x) = x^n\Id_N + \sum_{j = 0}^{n-1}\aa_{n,j}(s)x^j.
\end{gather*}

\begin{Proposition}[{\cite[Corollary 2.12]{OPRH}}]\label{prop1} All coefficients $\aa_{n,j}(s)$ of the monic matrix-valued orthogonal polynomials are real matrices. Additionally we have
 \begin{gather}
 \YY^{(n)}_{-1}(s) = \left( \begin{matrix}
 \aa_{n,n-1}(s) & -\dfrac{1}{2 \pi i} \gg_n^{-1}(s) \vspace{1mm}\\
 -2\pi i \gg_{n-1}(s) & -\aa^*_{n,n-1}(s)
 \end{matrix}\right). \label{Y-1}
 \end{gather}
\end{Proposition}
For the next proposition, we introduce the quantities $\pp_n(s),\qq_n(s)$, defined by
\begin{gather}
 \pp_n(s) := \cC(\hPP_n\WW)(s;0)\hPP_n^*(s;0), \nonumber\\
 \qq_n(s) := 2\pi i \gg_{n-1}(s) \hPP_{n-1}(s;0)\hPP_n^{-1}(s;0).\label{defpq}
\end{gather}
\begin{Proposition}\label{prop2}
 The matrix $\QQ^{(n)}(s) = \YY^{(n)}(s;0)$ and its inverse $\QQ^{(-n)}(s)$ can be written as
 \begin{gather}
 \QQ^{(n)}(s) = \left( \begin{matrix}
 \Id_N &\pp_n(s) \\
 -\qq_n(s) & \Id_N + \qq_n(s)\pp^*_n(s)
 \end{matrix}\right)
 \left( \begin{matrix}
 \hPP_n(s;0) & 0 \\
 0 & \hPP_n^{-*}(s;0)
 \end{matrix}\right), \label{QQn1}\\
 \QQ^{(-n)}(s) = \left( \begin{matrix}
 \hPP_n^{-1}(s;0) & 0 \\
 0 & \hPP_n^{*}(s;0)
 \end{matrix}\right)
 \left( \begin{matrix}
 \Id_N + \pp_n(s)\qq^*_n(s) & \pp_n^*(s) \\
 -\qq_n^*(s) & \Id_N
 \end{matrix}\right). \label{QQn-1}
\end{gather}
\end{Proposition}
\begin{proof}The first formula \eqref{QQn1} is a direct computation of the definition of $\YY^{(n)}(s;z)$ evaluated at $z=0$ (see \eqref{solutionRHP}), the definition of $\pp_n(s), \qq_n(s)$ in \eqref{defpq} and the Liouville--Ostrogradski formula (see \cite[Proposition~2.8]{OPRH}), i.e.,
\begin{gather}\label{LOF}
2\pi i\gg_{n-1}(s)\big(\hPP_{n-1}(s;z)\cC(\WW\hPP_n^*)(s;z)-\cC(\hPP_{n-1}\WW)(s;z)\hPP_{n}^*(s;z)\big)=\Id_N,
\end{gather}
evaluated at $z=0$. The second formula \eqref{QQn-1} is a consequence of the definition of the inverse of the Riemann--Hilbert problem~\eqref{solutionRHP}, which can be found in formula~(2.16) of~\cite{OPRH}. Indeed,
\begin{gather*}
 \YY^{(-n)}(s;z) = \left( \begin{matrix} -2\pi i \cC(\WW\hPP_{n-1}^*)(s;z)\gg_{n-1}(s) & -\cC(\WW\hPP_n^*)(s;z)\\
 2\pi i \hPP_{n-1}^*(s;z)\gg_{n-1}(s) & \hPP_n^*(s;z)
 \end{matrix}\right).
\end{gather*}
Therefore $\QQ^{(-n)}(s)=\YY^{(-n)}(s;0)$ is a consequence of the definition of $\pp_n(s)$, $\qq_n(s)$ in \eqref{defpq} and the Hermitian transpose version of the Liouville--Ostrogradski formula~\eqref{LOF}.
\end{proof}

A simple but important consequence of this proposition is the corollary below.
\begin{Corollary}\label{cor1}For any $n \geq 0$ the matrices $\pp_n(s)$ and $\qq_n(s)$ in~\eqref{defpq} are skew-Hermitian. Therefore $\QQ^{(n)}(s)$ and $\QQ^{(-n)}(s)$ can be rewritten in the following simplified way
\begin{gather*}
 \QQ^{(n)}(s) = \left( \begin{matrix}
 \Id_N &\pp_n(s) \\
 -\qq_n(s) & \Id_N - \qq_n(s)\pp_n(s)
 \end{matrix}\right)
 \left( \begin{matrix}
 \hPP_n(s;0) & 0 \\
 0 & \hPP_n^{-*}(s;0)
 \end{matrix}\right), \\ 
 \QQ^{(-n)}(s) = \left( \begin{matrix}
 \hPP_n^{-1}(s;0) & 0 \\
 0 & \hPP_n^{*}(s;0)
 \end{matrix}\right)
 \left( \begin{matrix}
 \Id_N - \pp_n(s)\qq_n(s) & -\pp_n(s) \\
 \qq_n(s) & \Id_N
 \end{matrix}\right). 
\end{gather*}
\end{Corollary}
\begin{proof} It suffices to compute the quantity $\QQ^{(n)}(s)\QQ_n^{(-n)}(s) = \Id_N$ using the equations~\eqref{QQn1} and~\eqref{QQn-1} above, which gives $\pp_n(s)+\pp_n^{*}(s)=0$ for the (block) entry $(1,2)$ and $\qq_n(s)+\qq_n^{*}(s)=0$ for the (block) entry~$(2,1)$. Using this, $\QQ^{(-n)}(s)\QQ_n^{(n)}(s) = \Id_N$ holds immediately.
\end{proof}

We are now ready to state the main result of this Section, giving a Lax system associated to our matrix-valued orthogonal polynomials. In the following we denote by $\aalpha_n(s)$ and $\bbeta_n(s)$ the recursion coefficients associated to our family of matrix-valued orthogonal polynomials, such that, for every $n \geq 1$,
\begin{gather}\label{recoef}
 x\hPP_n(s;x) = \hPP_{n+1}(s;x) + \aalpha_n(s)\hPP_n(s;x) + \bbeta_{n}(s)\hPP_{n-1}(s;x),
\end{gather}
and we recall that we have the identity
\begin{gather}\label{betn}
 \bbeta_n(s) = \gg_n^{-1}(s)\gg_{n-1}(s).
\end{gather}
Moreover, we define
\begin{gather*}\ssigma_3 := \left( \begin{matrix}
 \Id_N & 0 \\
 0 & -\Id_N
 \end{matrix}\right).\end{gather*}
\begin{Theorem}Let $\YY^{(n)}(s;z)$ be the solution of the Riemann--Hilbert problem \eqref{jumpcondition}--\eqref{asymp0} and
 \begin{gather}\label{RRi}
 \RR(s;z) := \left( \begin{matrix}
 {\rm e}^{-\frac{1}2(z + \frac{s}z)}z^{\frac{\alpha}2}\TT(s;z) & 0 \\
 0 & {\rm e}^{\frac{1}2(z + \frac{s}{z})}z^{-\frac{\alpha}2}\TT^{-*}(s;z)
 \end{matrix}\right).
\end{gather}
Then the function
\begin{gather}\label{ssi}
 \Ppsi^{(n)}(s;z) := \YY^{(n)}(s;z)\RR(s;z)
\end{gather}
satisfies the following equations
\begin{gather}
\frac{\pa}{\pa z}\Ppsi^{(n)}(s;z) = \left(-\frac{1}2 \ssigma_3 + \frac{\AA_{-1}^{(n)}(s)}{z} + \frac{\AA_{-2}^{(n)}(s)}{z^2} \right)\Ppsi^{(n)}(s;z), \nonumber\\
\frac{\pa}{\pa s}\Ppsi^{(n)}(s;z) = -\frac{\AA_{-2}^{(n)}(s)}{sz}\Ppsi^{(n)}(s;z), \nonumber\\
\Ppsi^{(n+1)}(s;z) = \UU^{(n)}(s;z) \Ppsi^{(n)}(s;z),\label{Laxtriple}
 \end{gather}
where the matrices $\AA_{-1}^{(n)}(s),\AA_{-2}^{(n)}(s)$ and $\UU^{(n)}(s;z)$ are explicitly given, in terms of the matrix-valued orthogonal polynomials, by
\begin{gather}
\AA_{-1}^{(n)}(s) = \left(\begin{matrix}
 (n + \alpha/2)\Id_N + \BB & -\dfrac{1}{2 \pi i}\gg_n^{-1}(s)\\
 2 \pi i \gg_{n-1}(s) & -(n + \alpha/2)\Id_N - \BB^*
 \end{matrix}\right),\nonumber \\
\AA_{-2}^{(n)}(s) = \left(\begin{matrix}
 \dfrac{s}{2}(\Id_N - 2\pp_n(s)\qq_n(s)) & -s\pp_n(s) \\
 -s\qq_n(s)(\Id_N - \pp_n(s)\qq_n(s)) & -\dfrac{s}{2}(\Id_N - 2\qq_n(s)\pp_n(s))
 \end{matrix}\right), \nonumber\\
\UU^{(n)}(s;z) = \left(\begin{matrix}
 z\Id_N - \aalpha_n(s) & \dfrac{1}{2 \pi i}\gg_n^{-1}(s)\\
 -2\pi i \gg_n(s) & 0
 \end{matrix}\right). \label{A-2m}
\end{gather}
\end{Theorem}
\begin{proof}The third equation in \eqref{Laxtriple} as well as the expression of $\UU^{(n)}(s;z)$ does not depend on the particular weight we choose and their derivation can be found in \cite[Theorem~2.16]{OPRH}. For the first equation in \eqref{Laxtriple}, observe that $\Ppsi^{(n)}(s;z)$ satisfies a Riemann--Hilbert problem with jumps independent on $z$ and $s$. More precisely, its jumps are localised on the positive real axis (because of the original jumps of $\YY^{(n)}(s;z)$) and, for general $\alpha>0$, on the negative real axis, because of the determination of $z^\alpha$. Hence, the expression $(\pa_z \Ppsi^{(n)}(s;z))\Ppsi^{(-n)}(s;z)$ is analytic on $\C\mathbb{P}^{1}\setminus\{0,\infty\}$, and studying its behaviour at the singular points we can conclude that
 \begin{gather}\label{masterthm}
 (\pa_z \Ppsi^{(n)}(s;z))\Ppsi^{(-n)}(s;z) = -\frac{1}2 \ssigma_3 + \frac{\AA_{-1}^{(n)}(s)}{z} + \frac{\AA_{-2}^{(n)}(s)}{z^2},
\end{gather}
for some matrices $\AA_{-1}^{(n)}(s), \AA_{-2}^{(n)}(s)$ to be determined. For $\AA^{(n)}_{-1}(s)$, we develop the left hand side of~\eqref{masterthm} at infinity. Because of the particular shape of $\RR(s;z)$, we conclude that
 \begin{gather*}
 \AA^{(n)}_{-1}(s) = \left(\begin{matrix}
 (n + \alpha/2)\Id_N + \BB & 0\\
 0 & -(n + \alpha/2)\Id_N - \BB^*
 \end{matrix}\right) + \frac{1}2\big[\ssigma_3, \YY_{-1}^{(n)}(s)\big],
 \end{gather*}
 where $[\cdot,\cdot]$ denotes the standard commutator operator. This last equation together with the Proposition \ref{prop1} gives the desired form of $\AA^{(n)}_{-1}(s)$. For $\AA^{(n)}_{-2}(s)$, we expand the equation \eqref{masterthm} around zero and we obtain that
 \begin{gather*}
 \AA^{(n)}_{-2}(s) =\frac{s}2 \QQ^{(n)}(s) \ssigma_3 \QQ^{(-n)}(s),
 \end{gather*}
and this equation, together with the Corollary~\ref{cor1}, gives the precise form of~$\AA^{(n)}_{-2}(s)$. The second equation in~\eqref{Laxtriple} is deduced similarly, but it is simpler. Indeed, the expression
\begin{gather*} (\pa_s \Ppsi^{(n)}(s;z))\Ppsi^{(-n)}(s;z)
\end{gather*} has a singularity only at zero, and expanding around this point one finds
\begin{gather*}
(\pa_s \Ppsi^{(n)}(s;z))\Ppsi^{(-n)}(s;z) = -\frac{1}{2z} \QQ^{(n)}(s) \ssigma_3 \QQ^{(-n)}(s) = -\frac{\AA_{-2}^{(n)}(s)}{sz} . \tag*{\qed}
\end{gather*}\renewcommand{\qed}{}
\end{proof}

\section{The Lax equations and some monodromy identities}

Let us denote
 \begin{gather}\label{Abs}
 \cAA^{(n)}(s;z) := -\frac{1}2 \ssigma_3 + \frac{\AA_{-1}^{(n)}(s)}{z} + \frac{\AA_{-2}^{(n)}(s)}{z^2}, \qquad \cBB^{(n)}(s;z) := -\frac{\AA_{-2}^{(n)}(s)}{sz}.
 \end{gather}
The compatibility conditions between the equations in \eqref{Laxtriple} give rise to three matrix equations of the form
 \begin{gather}
 \frac{\pa}{\pa s} \cAA^{(n)}(s;z) - \frac{\pa}{\pa z} \cBB^{(n)}(s;z) + \big[ \cAA^{(n)}(s;z) , \cBB^{(n)}(s;z) \big] = 0, \label{PIII}\\
 \frac{\pa}{\pa z} \UU^{(n)}(s;z) + \UU^{(n)}(s;z)\cAA^{(n)}(s;z) - \cAA^{(n+1)}\UU^{(n)}(s;z) = 0, \label{dPIII}\\
 \frac{\pa}{\pa s} \UU^{(n)}(s;z) + \UU^{(n)}(s;z)\cBB^{(n)}(s;z) - \cBB^{(n+1)}(s;z)\UU^{(n)}(s;z) = 0. \label{Toda}
 \end{gather}
The equation \eqref{PIII} describes the isomonodromic deformation of the first equation in \eqref{Laxtriple} with respect to the parameter $s$ and will give rise to equations of Painlev\'e type; analogously the equation~\eqref{dPIII} describes the isomonodromic deformation of the first equation in \eqref{Laxtriple} with respect to the discrete parameter~$n$, and hence it corresponds to equations of discrete Painlev\'e type. Finally, the equation \eqref{Toda} describes the compatibility between the two deformations (discrete and continuous), and it will give equations of Toda type. In the following three propositions we will write down explicitly all these equations. It is more convenient, for what follows, to use, instead of the variables $\pp_n(s)$ and $\qq_n(s)$, the variables $\aa_n(s)$ and $\bb_n(s)$ defined by
\begin{gather}\label{defab}
 \aa_n(s) := 2 \pi i s \pp_n(s) \gg_n(s), \qquad \bb_n(s) := s \pp_n(s)\qq_n(s),
\end{gather}
as well as the quantities
\begin{gather}\label{defBn}
 \BB_n = \BB_n(s) := \gg_n(s) \BB \gg_n^{-1}(s), \qquad \hat{\BB}_n = \hat{\BB}_n(s) := \hPP_n(s;0) \BB \hPP_n^{-1}(s;0).
\end{gather}
Observe that the coefficient $\AA_{-2}^{(n)}(s)$ in \eqref{A-2m} can be written (with the new notation) in the following way:
\begin{gather}\label{A-2m2}
\AA_{-2}^{(n)}(s)=\left(\begin{matrix}
 \dfrac{s}{2}\Id_N-\bb_n(s) & -\dfrac{1}{2\pi i}\aa_n(s)\gg_n^{-1}(s) \vspace{1mm}\\
 -2\pi i\gg_n(s)\aa_n^{-1}(s)\bb_n(s)(s\Id_N-\bb_n(s)) & -\dfrac{s}{2}\Id_N+\bb_n^*(s)
 \end{matrix}\right)
\end{gather}
\begin{Remark}\label{rem1}
Observe that, using \eqref{defpq} and \cite[Lemma~2.19]{OPRH}, the coefficients $\aa_n(s)$ and $\bb_n(s)$ can be written in the following way
\begin{gather}
\label{ans1} \aa_n(s) = s\left(\int_0^{\infty}\frac{\hPP_n(s;y)\WW(s;y)\hPP_n^*(s;y)}{y}\d y\right)\gg_n(s),\\
\label{bns1} \bb_n(s) = s\left(\int_0^{\infty}\frac{\hPP_n(s;y)\WW(s;y)\hPP_{n-1}^*(s;y)}{y}\d y\right)\gg_{n-1}(s).
\end{gather}
These definitions are the matrix-valued versions of the coefficients $a_n$ and $b_n$ that appear in \cite[Lemma~2]{ChenIts}.
\end{Remark}

 In the following, in order to make the equations more readable, we suppress the (implicit) dependence on~$s$.
\begin{Proposition}
 The compatibility condition \eqref{PIII} is equivalent to the following list of equations:
 \begin{gather}
 s\dot{\gg}_n = \gg_n\aa_n, \label{P1}\\
 s\dot{\gg}_{n-1} = -\gg_n\aa_n^{-1}\bb_n(s\Id_N - \bb_n), \label{P2}\\
 s\dot{\aa}_n= (2n + \alpha + 1)\aa_n + \aa_n^2 + \BB\aa_n + \aa_n \BB^*_n - s\Id_N + \bb_n + \gg_n^{-1}\bb_n^*\gg_n, \label{P3}\\
 s\dot{\bb}_n = \bb_n + [\BB, \bb_n] - \aa_n^{-1}\bb_n(s\Id_N - \bb_n ) - \aa_n\bbeta_n \label{P4}.
\end{gather}
\end{Proposition}
\begin{proof}
The proof of this proposition, as the following one, is just by straightforward computations. More precisely, we have, using the new definition of $\AA_{-2}^{(n)}(s)$ in~\eqref{A-2m2}, that the coefficient~$z^{-1}$ of the entry $(1,2)$ of~\eqref{PIII} gives~\eqref{P1}; the coefficient $z^{-1}$ of the entry $(2,1)$ of~\eqref{PIII} gives~\eqref{P2}; the coefficient $z^{-2}$ of the entry $(1,2)$ of~\eqref{PIII} gives~\eqref{P3} and the coefficient $z^{-2}$ of the entry $(1,1)$ of \eqref{PIII} gives~\eqref{P4}. The other entries give linear combinations of the ones already used.
\end{proof}

\begin{Proposition} The compatibility condition \eqref{dPIII} is equivalent to the following list of equations:
 \begin{gather}
 \aalpha_n = (2n + \alpha + 1)\Id_N + \aa_{n} + \BB + \BB^*_n, \label{dP1}\\
 s\Id_N - \aalpha_n\aa_n = \bb_{n+1} + \gg_n^{-1}\bb_n^*\gg_n, \label{dP2} \\
 \bb^2_{n+1} - s\bb_{n+1} = \aa_{n+1}\bbeta_{n+1}\aa_n, \label{dP3}\\
 \bbeta_{n+1} - \bbeta_n = \aalpha_{n} + \bb_{n+1} - \bb_n + [\BB, \aalpha_n ], \label{dP4}\\
 \aa_{n+1}\bbeta_{n+1} - \bbeta_{n}\aa_{n-1} = \aalpha_n\bb_n - \bb_{n+1}\aalpha_n. \label{dP5}
\end{gather}
\end{Proposition}
\begin{proof}
\eqref{dP1} is given by the term in $z^{-1}$ in the entry $(1,2)$ of \eqref{dPIII}; \eqref{dP2} is given by the term in $z^{-2}$ in the entry $(1,2)$ of \eqref{dPIII}; \eqref{dP3} is given by the term in $z^{-2}$ in the entry $(2,2)$ of \eqref{dPIII}; \eqref{dP4} is given by the term in $z^{-1}$ in the entry $(1,1)$ of \eqref{dPIII} and \eqref{dP5} is given by the term in $z^{-2}$ in the entry $(1,1)$ of \eqref{dPIII}. All the other entries give either linear combinations of the relations above or trivial relations.
\end{proof}

The compatibility condition \eqref{Toda} gives no new equations, except for the time derivative of~$\aalpha_n(s)$, reading
\begin{gather*}
 s\dot{\aalpha}_n = \bb_n - \bb_{n+1},
\end{gather*}
(this is given by the coefficient $z^0$ in the $(1,1)$ entry of \eqref{Toda}). Using \eqref{dP4}, \eqref{betn} and \eqref{P1} we get the couple of first-order differential equations
\begin{gather*}
\nonumber s\dot{\aalpha}_n= \aalpha_n-\bbeta_{n+1}+\bbeta_n+ [\BB, \aalpha_n ],\\
\nonumber s\dot{\bbeta}_n = \bbeta_n\aa_{n-1}-\aa_n\bbeta_n,
\end{gather*}
which, using \eqref{dP1}, gives a couple of differential equations satisfied by the coefficients~$\aalpha_n$ and~$\bbeta_n$ of the three-term recurrence relation~\eqref{recoef}. These equations, in the scalar setting, are also known as \emph{Toda equations}.

Another interesting relation is combining \eqref{dP1}, \eqref{dP4} and performing a telescopic sum to get
\begin{gather}\label{dP4e}
\bbeta_n=n((n+\alpha)\Id_N+\BB)+\bb_n+\sum_{k=0}^{n-1}\aa_k+\BB_k^*+ [\BB, \aa_k+\BB_k^* ].
\end{gather}

In the scalar case, there is one more important identity that can be deduced by the (scalar analogues of the) equations \eqref{dP1}--\eqref{dP5}; it is the formal monodromy identity stated in \cite[Lemma~3]{ChenIts}. Nevertheless the proof, in the matrix case, cannot be repeated in the same way and we have to use (as suggested by the name given to the identity) some relations between the expansion of $\Ppsi^{(n)}(s;z)$ at infinity and at zero.

\begin{Proposition}\label{propg}
\begin{gather}
\aa_n\bbeta_n = s\big(n\Id_N+\BB+\aa_n^{-1}\bb_n-\hat{\BB}_n\big)-\bb_n ((2n+\alpha)\Id_N+\BB ) \nonumber\\
\hphantom{\aa_n\bbeta_n =}{} -\bb_n\aa_n^{-1}\bb_n-\aa_n\BB_n^*\aa_n^{-1}\bb_n.\label{formalmonodromy}
\end{gather}
\end{Proposition}
\begin{proof}We recall that, in the Section above, the explicit expression of $\AA_{-1}^{(n)}(s)$ had been obtained computing the coefficient $z^{-1}$ in the expansion of $\big(\partial_z \Ppsi^{(n)}(s;z)\big)\Ppsi^{(-n)}(s;z)$ around infinity, and using \eqref{Y-1}. Of course, one can also compute the expansion around zero of the same expression. The coefficient $z^{-1}$, in this case, will give the identity
 \begin{gather}\label{newA-1}
 \QQ^{(n)}(s)
 \left[ \left( \begin{matrix}
 \dfrac{\alpha}2\Id_N + \BB & 0\\
 0 & -\dfrac{\alpha}2\Id_N - \BB^*
 \end{matrix}\right) +
 \frac{s}2 \big[\YY_1^{(n)}(s),\ssigma_3\big] \right]
 \QQ_n^{(-n)}(s) = \AA_{-1}^{(n)}(s).
 \end{gather}
Bringing the conjugation by $\QQ^{(n)}(s)$ on the right hand side of~\eqref{newA-1}, the $(1,1)$ entry of the resulting equation will give exactly the equation~\eqref{formalmonodromy}.
\end{proof}

As a consequence of \eqref{formalmonodromy}, and using \eqref{dP3}, we have the identity
\begin{gather*}
\aa_n\bbeta_n+\bbeta_n\aa_{n-1}=s(n\Id_N+\BB-\hat{\BB}_n)-\bb_n((2n+\alpha)\Id_N+\BB)\\
\hphantom{\aa_n\bbeta_n+\bbeta_n\aa_{n-1}=}{} -\aa_n\BB_n^*\aa_n^{-1}\bb_n-\bb_n\aa_n^{-1}\bb_n+\aa_n^{-1}\bb_n^2,
\end{gather*}
which can be viewed as a matrix-valued version of formula~(2.12) of~\cite{ChenIts}.

\begin{Remark}
Symmetrically, one can also compute $\AA_{-2}^{(n)}(s)$ using the expansion of $\Ppsi^{(n)}(s;z)$ around infinity. In this case, the following equation is obtained:
 \begin{gather}
\frac{s}2\Id_N + \frac{1}2\big[\ssigma_3,\YY_{-2}^{(n)}\big] + \left[ \YY_{-1}^{(n)}, \left( \begin{matrix}
 (n+\alpha/2)\Id_N + \BB & 0\\
 0 & -(n+\alpha/2)\Id_N - \BB^*
 \end{matrix}\right) \right]\nonumber\\
\qquad{} + \frac{1}2 \big[\YY_{-1}^{(n)},\ssigma_3\big]\YY_{-1}^{(n)} - \YY_{-1}^{(n)} = \AA_{-2}^{(n)}. \label{newA-2}
 \end{gather}
The entry $(1,1)$ of the relation above gives the equation
\begin{gather*}
 \bbeta_{n} - \bb_n = \aa_{n,n-1} + [ \BB, \aa_{n,n-1} ],
\end{gather*}
which is equivalent to \eqref{dP4}, since by definition $\aalpha_n(s) = \aa_{n+1,n}(s) - \aa_{n,n-1}(s).$
\end{Remark}
\begin{Remark}The diagonal entries of \eqref{newA-1} and \eqref{newA-2} give, respectively, the off--diagonal block elements of $\YY_1^{(n)}(s)$ and $\YY_{-2}^{(n)}(s)$. One can continue on this direction and:
\begin{enumerate}\itemsep=0pt
\item[a)] compute the terms in $z^{k}$, $k \geq 0$ of the expansion of $\big(\partial_s \Ppsi^{(n)}(s;z)\big)\Ppsi^{(-n)}(s;z)$ around zero,
\item[b)] compute the terms in $z^{-k}$, $k \geq 3$ of the same expression around zero.
\end{enumerate}
In this way, all elements $\YY_{\pm k}$ in the expansions \eqref{asympinfty} and \eqref{asymp0} can be recursively computed just in function of the entries of $\QQ^{(n)}(s)$ and $\YY_{-1}^{(n)}(s)$.
\end{Remark}

\section{The non-commutative Toda and Painlev\'e systems}
In the scalar case, Chen and Its managed to give a closed system of two first-order difference and differential equations for the two variables $\aa_n$, $\bb_n$. This system is immediately seen to be equivalent to the Painlev\'e III equation for the variable~$\aa_n$ (see \cite[Lemma~5 and Theorem~1]{ChenIts}). Here we prove that, in the matrix case, it is possible to give a system (of the first-order) of difference and differential equations for the \emph{four} variables $\aa_n$, $\bb_n$, $\BB_n$, $\hat{\BB}_n$, where the presence of the last two variables is a clear consequence of the presence of the non-commutative constant matrix~$\BB$. Additionally, one can further reduce the system to a system of second-order difference and differential equations for the two variables~$\aa_n$, $\bb_n$. We need first the following lemma.

\begin{Lemma}For any $n \geq 0$, the following equation holds:
 \begin{gather}\label{bb*}
 \gg_n^{-1}(s)\bb^*_n(s)\gg_n(s) = \aa_n^{-1}(s)\bb_n(s)\aa_n(s).
 \end{gather}
 Moreover, the derivative of $\hPP_n(s;0)$ can be expressed as
 \begin{gather}\label{Pdot}
 s\dot{\hPP}_n(s;0)\hPP_n^{-1}(s;0) = n\Id_N + \BB - \hat{\BB}_n(s)+ \aa_n^{-1}(s)\bb_n(s).
 \end{gather}
\end{Lemma}
\begin{proof}For the first \eqref{bb*}, using \eqref{defab} and the fact that $\pp_n$ and $\qq_n$ are skew-Hermitian, we have that
\begin{gather*}
\aa_n\gg_n^{-1}\bb^*_n\gg_n=2\pi is\pp_n s\qq_n^*\pp_n^*\gg_n=(s\pp_n\qq_n)(2\pi is\pp_n\gg_n)=\bb_n\aa_n.
\end{gather*}
For the second \eqref{Pdot} we first observe that, expanding $\big(\partial_s \Ppsi^{(n)}(s;z)\big)\Ppsi^{(-n)}(s;z)$ around zero, the constant term in $z$ should give zero. This gives the equation
\begin{gather*}
 \dot{\QQ}^{(n)}(s) = \frac{1}2 \QQ^{(n)}(s)\big[\YY_1^{(n)}(s), \ssigma_3\big].
\end{gather*}
On the other hand \eqref{newA-1} can be rewritten as
\begin{gather*}
 \frac{s}2\QQ^{(n)}(s) \big[\YY_1^{(n)}(s), \ssigma_3\big] = \AA_{-1}^{(n)}(s)\QQ^{(n)}(s) - \QQ^{(n)}(s)\left( \begin{matrix}
 \dfrac{\alpha}2\Id_N + \BB & 0\\
 0 & -\dfrac{\alpha}2\Id_N - \BB^*
 \end{matrix}\right),
\end{gather*}
so that we get
\begin{gather*}
 s\dot{\QQ}^{(n)}(s) = \AA_{-1}^{(n)}(s)\QQ^{(n)}(s) - \QQ^{(n)}(s)\left( \begin{matrix}
 \dfrac{\alpha}2\Id_N + \BB & 0\\
 0 & -\dfrac{\alpha}2\Id_N - \BB^*
 \end{matrix}\right).
\end{gather*}
The equation \eqref{Pdot} is just the $(1,1)$ entry of this last identity.
\end{proof}

Now let us give the two main results of this paper, namely a couple of non-linear non-com\-mutative second-order difference and differential equations for the coefficients~$\aa_n(s)$ and~$\bb_n(s)$. We start first with the discrete version. In the formulas below, again, we suppress the implicit dependence on~$s$.

\begin{Proposition} The variables $\aa_n$, $\bb_n$, $\BB_n$, $\hat{\BB}_n$ defined in \eqref{defab} and \eqref{defBn} satisfy the $($closed$)$ system of first-order difference equations
\begin{gather}
\aa_{n-1}\bb_{n} + \bb_{n-1}\aa_{n-1} = s\aa_{n-1} - (2n + \alpha - 1)\aa_{n-1}^2 - \aa_{n-1}^3 -\aa_{n-1}( \BB +\BB^*_{n-1})\aa_{n-1},\nonumber\\
\bb_n^2 - s\bb_n = \big(s\big(n\Id_N + \BB-\hat{\BB}_n+\aa_n^{-1}\bb_n\big) - \bb_n\big(2n + \alpha + \BB+\aa_n^{-1}\bb_n\big) - \aa_n\BB^*_n\aa_n^{-1}\bb_n \big)\aa_{n-1}, \nonumber\\
\aa_n \BB_n^*\aa_n^{-1}\bb_n(s - \bb_n) = \bb_n(s - \bb_n)\aa_{n-1}^{-1}\BB_{n-1}^*\aa_{n-1},\nonumber\\
\hat{\BB}_n(s - \bb_n) = (s - \bb_n)\aa_{n-1}^{-1}\hat{\BB}_{n-1}\aa_{n-1}.\label{finalsystemdiscrete}
\end{gather}
\end{Proposition}
\begin{proof}For the first one, plug \eqref{dP1} into \eqref{dP2}, multiply on the left by $\aa_n$ and use \eqref{bb*}. The second is a consequence of plugging \eqref{formalmonodromy} into \eqref{dP3}. For the other two equations we start observing that
\begin{gather*}
 \BB_n = \gg_n\BB\gg_n^{-1}= \gg_n\gg_{n-1}^{-1}\BB_{n-1}\gg_{n-1}\gg_n^{-1}= \bbeta_n^{-*}\BB_{n-1}\bbeta_n^{*},
\end{gather*}
which gives
\begin{gather}\label{prethird}
 \bbeta_{n}^*\BB_{n} = \BB_{n-1}\bbeta_{n}^*.
\end{gather}
Analogously, using the definition of $\hat{\BB}_n$ and
\begin{gather*}
\hPP_n(s;0)\hPP_{n-1}^{-1}(s;0) = \bb_n^{-1}(s)\aa_n(s)\bbeta_n(s),
\end{gather*}
we obtain
\begin{gather}\label{prefourth}
 \hat{\BB}_{n}\bb_{n}^{-1}\aa_{n}\bbeta_{n} = \bb_{n}^{-1}\aa_{n}\bbeta_{n}\hat{\BB}_{n-1}.
\end{gather}
Then the last two equations in \eqref{finalsystemdiscrete} are obtained from \eqref{prethird} and~\eqref{prefourth} where~$\bbeta_n$ had been written in function of~$\aa_n$,~$\aa_{n-1}$ and~$\bb_n$ using~\eqref{dP3}.
 \end{proof}

The system above can be reduced to a system of higher order just for the two variab\-les~$\aa_n(s)$,~$\bb_n(s)$. It is practical, for the following, to introduce the following two quantities:
\begin{gather*}
 \CC_n := s\aa_n^{-1} - \aa_n - \aa_n\BB\aa_n^{-1} - \aa_n\bb_{n+1}\aa_n^{-2} - \bb_n\aa_n^{-1}, \\
 \DD_n := (s\Id_N - \bb_n)\big(\BB + \bb_n\aa_{n-1}^{-1}\big) + \big(\Id_N + \aa_n + \aa_n\BB\aa_n{-1}+\aa_n\bb_{n+1}\aa_n^{-2}\big)\bb_n.
\end{gather*}
\begin{Theorem}\label{discreteThm} The variables $\aa_n,\bb_n$ satisfy the $($closed$)$ difference system
\begin{gather}
 \CC_n \bb_n(s\Id_n - \bb_n) - \bb_n(s\Id_N - \bb_n)\aa_{n-1}^{-2}\CC_{n-1}\aa_{n-1}^2 = 2\bb_n(s\Id_N - \bb_n),\nonumber\\
 \DD_n(s\Id_N - \bb_n) - (s\Id_N - \bb_n)\aa_{n-1}^{-1}\DD_{n-1}\aa_{n-1}=s(\bb_n - s\Id_N).\label{discreteanbn}
\end{gather}
\end{Theorem}
\begin{proof}We start with the equation for $\CC_n$. Using the first relation in \eqref{finalsystemdiscrete} we obtain
\begin{gather*}
 \aa_n\BB_n^*\aa_n = s\aa_n - (2n + \alpha + 1)\aa_n^2 - \aa_n^3 - \aa_n\BB\aa_n - \aa_n\bb_{n+1} - \bb_n\aa_n ,
\end{gather*}
and we plug $\aa_n\BB_n^*\aa_n$ into the third equation of \eqref{finalsystemdiscrete}. For the second equation in \eqref{discreteanbn} we use the fact that
\begin{gather*}s\hat{\BB_n} = ns\Id_N + (s\Id_N - \bb_n)\BB - \big(\bb_n^2 - s\bb_n\big)\aa_{n-1}^{-1} + \big(\Id_N + \aa_n + \aa_n\BB\aa_n^{n-1} + \aa_n\bb_{n+1}\aa_n^{-2}\big)\bb_n,\end{gather*}
and we plug this relation into the last one of \eqref{finalsystemdiscrete}. \end{proof}

From the previous equations we can compute all the coefficients $\aa_n$, $\bb_n$, $\BB_n$, $\hat{\BB}_n$, $\aalpha_n$ and~$\bbeta_n$ in terms only on~$\aa_0(s)$, $\BB$ and $\gg_0$. Indeed, initially we have $\bb_0=0$, $\BB_0=\gg_0\BB\gg_0^{-1}$ and $\hat{\BB}_0=\BB$. From~\eqref{dP1} we can compute~$\aalpha_n$ in terms of~$\aa_n$,~$\BB$ and $\BB_n$ and from~\eqref{dP4e} (or~\eqref{dP4}) we can compute $\bbeta_n$ in terms of $\bb_n, \BB$ and $\aa_k$, $\BB_k$, $k=0,\ldots, n-1$. Therefore it is enough to compute $\aa_n$, $\bb_n$, $\BB_n$, $\hat{\BB}_n$ in terms only on $\aa_0(s)$, $\BB$ and $\gg_0$. For that we iterate the following 4 steps:
\begin{enumerate}\itemsep=0pt
\item From the first equation in \eqref{finalsystemdiscrete} we can compute $\bb_1$ in terms of $\aa_0$, $\BB$ and $\gg_0$.
\item From the fourth equation in \eqref{finalsystemdiscrete} we can compute $\hat{\BB}_1$ in terms of $\bb_1$, $\aa_0$ and $\BB$.
\item From \eqref{prethird} we can compute $\BB_1$ in terms of $\bbeta_1$, $\BB$ and $\gg_0$ (observe here that $\bbeta_1$ is computed from \eqref{dP4e} in terms of $\bb_1$, $\aa_0$ and $\BB$).
\item From \eqref{dP3} we can compute $\aa_1$ in terms of $\bbeta_1$, $\aa_0$, $\bb_1$ and $\BB$.
\end{enumerate}
Actually $\gg_0$ can be avoided if we normalize the weight matrix $\WW$ such that $\gg_0=\Id_N$. Since $\BB$ is a fixed matrix, in order to compute $\aa_0(s)$ we use \eqref{ans1} for $n=0$, i.e.,
\begin{gather*}
\aa_0(s)=s\left(\int_0^{\infty}\frac{\WW(s;y)}{y}\d y\right)\left(\int_0^{\infty}\WW(s;y)\d y\right)^{-1}.
\end{gather*}
Since $\WW(s;x)=x^\alpha{\rm e}^{-x - s/x}\TT(x)\TT^*(x)$, where $\TT(x)$ is typically a matrix polynomial, we have that $\aa_0(s)$ can be computed using the very well known formula
\begin{gather*}
\int_0^{\infty}x^{\alpha+k-1}{\rm e}^{-x - s/x}\d x=2\big(\sqrt{s}\big)^{\alpha+k}K_{\alpha+k}\big(2\sqrt{s}\big),\qquad k\geq1,
\end{gather*}
where $K_\nu(z)$ is the MacDonald function of the second type.

Let us now study the continuous version, i.e., a couple of non-linear non-commutative second-order differential equations for the coefficients $\aa_n(s)$ and $\bb_n(s)$.

\begin{Proposition} The variables $\aa_n$, $\bb_n$, $\BB_n$, $\hat{\BB}_n$ defined in \eqref{defab} and \eqref{defBn} satisfy the $($closed$)$ system of first-order differential equations
 \begin{gather}
s\dot{\aa}_n = - s\Id_N + \bb_n+\big((2n + \alpha + 1)\Id_N + \BB + \aa_n+\aa_n^{-1}\bb_n\big)\aa_n + \aa_n\BB_n^*, \nonumber\\
s\dot{\bb}_n = \bb_n\big((2n + \alpha + 1)\Id_N + \BB+\aa_n^{-1}\bb_n\big) - s\big(2 \aa_n^{-1}\bb_n + n\Id_N+ \BB-\hat{\BB}_n\big) \nonumber\\
\hphantom{s\dot{\bb}_n =}{} + \aa_n^{-1}\bb_n^2+ [\BB,\bb_n] + \aa_n\BB^*_n\aa_n^{-1}\bb_n,\nonumber\\
 s\dot{\BB}_n = [\aa_n^*,\BB_n],\nonumber\\
 s\dot{\hat{\BB}}_n = \big[\aa_n^{-1}\bb_n + \BB, \hat{\BB}_n\big]. \label{finalsystem}
\end{gather}
\end{Proposition}
\begin{proof}The first equation is just a rewriting of the equation \eqref{P3} using \eqref{bb*}, while the second equation is obtained plugging \eqref{formalmonodromy} into \eqref{P4}. For the third equation, we observe that, from the definition~\eqref{defBn}, we have that
\begin{gather*}
\dot{\BB}_n = \big[\dot{\gg}_n\gg_n^{-1},\BB_n\big],
\end{gather*}
and then one has simply to use \eqref{P1} and $\gg_n\aa_n=\aa_n^*\gg_n$. For the last equation, the derivation is similar, but this time one has to use the equation \eqref{Pdot} instead of \eqref{P1}. \end{proof}

Let us obtain now a closed system of second-order differential equations for the (matrix) variables $\aa_n$ and $\bb_n$. First, from the first two equations of \eqref{finalsystem}, we can write $\BB_n$ and $\hat{\BB}_n$ in function of the variables $\aa_n$, $\bb_n$ and their derivatives. Indeed, the first equation gives
\begin{gather}\label{Bns}
\BB_n^*=\aa_n^{-1}\big(s\dot{\aa}_n +s\Id_N - \bb_n-\big((2n + \alpha + 1)\Id_N + \BB + \aa_n+\aa_n^{-1}\bb_n\big)\aa_n \big).
\end{gather}
If we substitute \eqref{Bns} into the second equation of \eqref{finalsystem}, and after some computations, we get
\begin{gather}\label{Bng}
\hat{\BB}_n=\dot{\bb}_n+n\Id_N+\BB+\aa_n^{-1}\bb_n-\dot{\aa}_n\aa_n^{-1}\bb_n+\frac{\aa_n\bb_n}{s}.
\end{gather}
Plugging these two relations into the last two equations of \eqref{finalsystem} gives the following Theorem.

\begin{Theorem}\label{continuousThm}
The variables $\aa_n$, $\bb_n$ defined in~\eqref{defab} satisfy the $($closed$)$ system of second-order differential equations
\begin{gather}
\ddot{\aa}_n= \dot{\aa}_n\aa_n^{-1}\dot{\aa}_n+\dot{\aa}_n\aa_n^{-1} -\frac{1}{s^2}\big(\big[\BB\aa_n,\aa_n\big]+\big[\aa_n^{-1}\bb_n,\aa_n^2\big]\big) +\frac{1}{s}\big(\dot{\bb}_n-\dot{\aa}_n+\aa_n^{-1}\dot{\bb}_n\aa_n\nonumber \\
\hphantom{\ddot{\aa}_n=}{}+\dot{\aa}_n\aa_n-\big(\dot{\aa}_n\aa_n^{-2} +\aa_n^{-1}\dot{\aa}_n\big)\bb_n\aa_n+\BB\dot{\aa}_n
 -\dot{\aa}_n\aa_n^{-1}\BB\aa_n+\big[\aa_n^{-1}\bb_n,\dot{\aa}_n\big]-\Id_N\big)\label{F1eq}
\end{gather}
and
\begin{gather}
\ddot{\bb}_n = \big(\dot{\aa}_n\aa_n^{-1}+\aa_n^{-1}\dot{\aa}_n\big)\aa_n^{-1}\bb_n+\dot{\aa}_n\aa_n^{-1}\dot{\bb}_n-\aa_n^{-1}\dot{\bb}_n +\frac{1}{s^2}\aa_n (\bb_n+ [\BB,\bb_n ] )\nonumber\\
\hphantom{\ddot{\bb}_n =}{} -\frac{1}{s}\big(\aa_n\dot{\bb}_n+ [\dot{\bb}_n,\BB ]+\dot{\aa}_n\aa_n^{-1}\big[\BB,\bb_n\big] -\aa_n^{-1}(\dot{\bb}_n\bb_n+\bb_n\dot{\bb}_n) \nonumber\\
\hphantom{\ddot{\bb}_n =}{}+ \big(\dot{\aa}_n\aa_n^{-2}+\aa_n^{-1}\dot{\aa}_n\aa_n^{-1}\big)\bb_n^2+\dot{\aa}_n\aa_n^{-1}\bb_n+\aa_n^{-1}\bb_n\big).\label{F2eq}
\end{gather}
\end{Theorem}
\begin{proof}The first equation \eqref{F1eq} follows after some long but straightforward computations plugging~\eqref{Bns} into the third equation of~\eqref{finalsystem}. For the second equation~\eqref{F2eq}, we plug~\eqref{Bng} into the fourth equation of~\eqref{finalsystem} and use again~\eqref{F1eq}.
 \end{proof}

The initial conditions for the couple of second-order differential equations \eqref{F1eq} and \eqref{F2eq} are given by $\aa_n(0)=\bb_n(0)=0$ and, differentiating \eqref{ans1} and \eqref{bns1} with respect to $s$, we get
\begin{gather}
\label{iniconda} \dot{\aa}_n(0) = \left(\int_0^{\infty}\frac{\hPP_n(y)\WW(y)\hPP_n^*(y)}{y}\d y\right)\left(\int_0^{\infty}\hPP_n(y)\WW(y)\hPP_n^*(y)\d y\right)^{-1},\\
\label{inicondb} \dot{\bb}_n(0) = \left(\int_0^{\infty}\frac{\hPP_n(y)\WW(y)\hPP_{n-1}^*(y)}{y}\d y\right)\left(\int_0^{\infty}\hPP_{n-1}(y)\WW(y)\hPP_{n-1}^*(y)\d y\right)^{-1}.
\end{gather}
Observe here that the monic matrix-valued orthogonal polynomials $\hPP_n(x)$ and the weight matrix $\WW(x)$ are evaluated at $s=0$. Since in this case $\WW(x)=x^\alpha{\rm e}^{-x}\TT(x)\TT^*(x)$, where $\TT(x)$ is typically a matrix polynomial, and $\hPP_n(x)$ can be expressed typically in terms of Laguerre polynomials (see for instance~\cite{DL}), it will be possible to compute~\eqref{iniconda} and~\eqref{inicondb} from the formula
\begin{gather*}
\int_0^{\infty}\hat{L}_n^{(a_1)}(x)\hat{L}_{m}^{(a_2)}(x)x^{\sigma-1}{\rm e}^{-x}\d x\\
\qquad{} =(-1)^{n+m}\Gamma(\sigma)(a_1+1)_n(a_2+1)_m\sum_{i=0}^n\sum_{j=0}^m\frac{(\sigma)_{i+j}(-n)_i(-m)_j}{(a_1+1)_i(a_2+1)_ji!j!},
\end{gather*}
for $\Re(\sigma)>0$, where here $\hat{L}_n^{(\alpha)}(x)$ denotes the monic Laguerre polynomial and $(a)_k$ is the Pochhammer symbol. A more general formula of this type can be found in~\cite{Erd}, where the right-hand part can be written in terms of Appell hypergeometric series.

\begin{Remark}If we assume that all the coefficients are scalar functions (denoted without boldfaced fonts), then we have that $\BB=\BB_n^*=\hat{\BB}_n=0$. From the first two equations of~\eqref{finalsystem} we can get $b_n$, $\dot{b}_n$ and $\ddot{b}_n$ in terms of~$a_n$ and $\dot{a}_n$ (just like in the scalar case, see formulas~(3.10) and~(3.11) of~\cite{ChenIts}). Once we substitute all these values in the second-order differential equation~\eqref{F1eq} we get the well-known version of the Painlev\'e~III equation
\begin{gather*}
\ddot{a}_n=\frac{(\dot{a}_n)^2}{a_n}-\frac{\dot{a}_n}{s}+(2n+\alpha+1)\frac{a_n^2}{s^2}+\frac{a_n^3}{s^2}+\frac{\alpha}{s}-\frac{1}{a_n}.
\end{gather*}
Additionally, if we make again all the substitutions and using this previous Painlev\'e III equation, it is easy to see that the second-order differential equation~\eqref{F2eq} also holds.
\end{Remark}

\section{Additional relations for special situations}\label{SEC4}

The differential relations obtained from the transformation \eqref{ssi} of the Riemann--Hilbert problem are not unique in the sense that we could have always introduced an unitary matrix-valued function~$\bm S(z)$ inside the factorization of the weight matrix~$\WW(s;z)$ of the form
\begin{gather*}
\WW(s;z) := z^\alpha{\rm e}^{-z - s/z}\TT(z)\bm S(z)\bm S^*(z)\TT^*(z).
\end{gather*}
Although the weight matrix is the same and the solution of the Riemann--Hilbert problem is unique, if we perform the transformation~\eqref{RRi} with $\TT\SS$ instead of~$\TT$ then we should expect a new Lax pair with additional information. This was already pointed out in \cite[Section~3]{OPRH}, where the authors applied this approach to simplify considerably the differential relations for some Hermite-type matrix-valued examples. The situation in the scalar case is pointless and it is not possible to get new relations from this approach (see \cite[Proposition~3.9]{OPRH}). Nevertheless in order to produce new interesting relations, we have to take an specific choice of the matrix~$\BB$.

Let us call $\bm\chi(z):=(\partial_z \bm S(z)) \bm S^*(z)$ the log-derivative of $\SS$ . Then the log-derivative of $\TT\SS$ is given by
\begin{gather*}
 \big(\partial_z ( \TT(z)\SS(z) )\big) \SS^*(z)\TT^{-1}(z) = \frac{\BB}z+z^{\BB}\bm\chi(z)z^{-\BB}.
\end{gather*}
Denoting $\HH(z):=z^{\BB}\bm\chi(z)z^{-\BB}$, if we manage to find a matrix-valued unitary function~$\SS$ such that $z\HH(z)$ is a matrix polynomial, then it will be possible to compute explicitly the new coefficient $\cAA^{(n)}(s;z)$ in~\eqref{Abs}, and we can get new compatibility conditions. In this case we have
\begin{gather*}
\HH(z)=\bm\chi(z)+\operatorname{ad}_{\BB}(\bm\chi(z))\log z+\operatorname{ad}^2_{\BB}(\bm\chi(z))\frac{\log^2z}{2}+\cdots=\sum_{k=0}^{\infty}\operatorname{ad}^k_{\BB}(\bm\chi(z))\frac{\log^kz}{k!},
\end{gather*}
where $\operatorname{ad}_{\bm X}(\bm Y)$ is the commutator $\big[\bm X,\bm Y\big]$. Here we define recursively $\operatorname{ad}^{n+1}_{\bm X}(\bm Y)\!=\!\operatorname{ad}_{\bm X}(\operatorname{ad}^n_{\bm X}(\bm Y))$ for $n\geq1$ with $\operatorname{ad}^0_{\bm X}(\bm Y)=\bm Y$. One possible choice was given in \cite[Section~6.2]{DG1}. Consider $\BB$ and $\BB_0$ satisfying $ [ \BB,\BB_0 ]=\BB_0$ and that $\BB^2+\alpha\BB-\BB_0$ is Hermitian. Then $\BB$ and $\BB_0$ have a~special structure. In this case, it was proven in \cite{DG1} that $\bm B=\bm Z\bm J\bm Z^{-1}$ and $\bm B_0=\bm Z\bm L\bm Z^{-1}$, where $\bm L$ and $\bm J$ are the nilpotent and diagonal matrices given by
\begin{gather*}
 \bm L =\sum_{k=1}^{N-1}\nu_k\bm E_{k,k+1},\qquad\nu_k\in\mathbb{C}\setminus\{0\}, \qquad \bm J=\sum_{k=1}^N(N-k)\bm E_{k,k},
\end{gather*}
and $\bm Z$ is the transformation matrix given by
\begin{gather*}
\bm Z=(z_{ij})_{i,j=1,\ldots N},\qquad z_{ij}=
 \begin{cases}
 0, & \hbox{if $i>j$;} \\
 1, & \hbox{if $i=j$;} \\
 \displaystyle\prod_{l=1}^{j-i} \frac{\nu_{i+l-1}}{c_{i+l}-c_i}, & \hbox{if $i<j$,}
 \end{cases}
\end{gather*}
where $c_i$, $i=1,\ldots,N,$ are the diagonal entries of $\bm J^2+\alpha \bm J$. Here $\bm E_{ij}$ is a matrix with~1 at entry~$(i,j)$ and~0 elsewhere. Observe now that $\bm B$ \emph{is not any matrix} and it only depends on $N-1$ free parameters. For instance, for $N=2$ we have only one free parameter and
\begin{gather*}
\BB=\left( \begin{matrix}
 1& -\dfrac{\nu_1}{\alpha+1}\\
 0&0
 \end{matrix}\right),\qquad \BB_0=\left( \begin{matrix}
 0& \nu_1\\
 0&0
 \end{matrix}\right).
\end{gather*}
The weight matrix $\WW(s;x)$ is given in this case by
\begin{gather*}
\WW(s;x)=x^\alpha{\rm e}^{-x - s/x}\left( \begin{matrix}
 x^2+\dfrac{\nu^2(x-1)^2}{(\alpha+1)^2}& -\dfrac{\nu(x-1)}{\alpha+1}\vspace{1mm}\\
 -\dfrac{\nu(x-1)}{\alpha+1}&1
 \end{matrix}\right),\qquad x\in[0,\infty),\qquad \alpha,s>0.
\end{gather*}

\looseness=-1 Let $\SS(z)=z^{i(\BB^2+\alpha\BB-\BB_0)}$. Observe that, since $\BB^2+\alpha\BB-\BB_0$ is Hermitian we have that~$\SS(z)$ is a unitary matrix-valued function. Therefore, the matrix-valued function~$\bm\chi(z)$ is then given by
\begin{gather*}
\bm\chi(z)=i\left(\frac{\BB^2+\alpha\BB-\BB_0}{z}\right).
\end{gather*}
A straightforward computation using that $[ \BB,\BB_0]=\BB_0$ gives
\begin{gather*}
\bm H(z)=z^{\BB}\bm\chi(z)z^{-\BB}=i\left(\frac{\BB^2+\alpha\BB}{z}-\BB_0\right).
\end{gather*}
Therefore, the first equation in the Lax triple \eqref{Laxtriple} can be written as
\begin{gather*}
 \frac{\pa}{\pa z}\Ppsi^{(n)}(s;z) = \big(\cAA^{(n)}(s;z)+\cAA^{(n)}_{\HH}(s;z) \big)\Ppsi^{(n)}(s;z),
\end{gather*}
where $\cAA^{(n)}(s;z)$ is defined in \eqref{Abs} and
\begin{gather*}
\cAA^{(n)}_{\HH}(s;z)=\HH_0^{(n)}+\frac{\HH_{-1}^{(n)}(s)}{z},
\end{gather*}
where
\begin{gather*}
\HH_0^{(n)}=-i\left(\begin{matrix}
 \BB_0 & 0\\
 0& \BB_0^*
 \end{matrix}\right)
\end{gather*}
and
\begin{gather*}
\HH_{-1}^{(n)}(s)=i\left(\begin{matrix}
 \BB^2+\alpha\BB+ [\BB_0,\aa_{n,n-1}(s)] & \dfrac{1}{2\pi i}\gg_{n}^{-1}(\BB_0^*-\LL_n(s))\vspace{1mm}\\
 -2\pi i(\BB_0^*-\bm L_{n-1}(s))\gg_{n-1}& \big(\BB^2+\alpha\BB+ [\BB_0,\aa_{n,n-1}(s) ]\big)^*
 \end{matrix}\right).
\end{gather*}
Here we are using the notation
\begin{gather*}
\bm L_n(s)=\gg_n(s)\BB_0\gg_n^{-1}(s).
\end{gather*}
From the first compatibility condition \eqref{PIII} we get a new relation
\begin{gather*}
\partial_s\cAA^{(n)}_{\HH}(s;z)+\big[\cAA^{(n)}_{\HH}(s;z),\cBB^{(n)}(s;z)\big]=0,
\end{gather*}
while form the second compatibility condition \eqref{dPIII} we get
\begin{gather*}
\UU^{(n)}(s;z)\cAA^{(n)}_{\HH}(s;z)=\cAA^{(n+1)}_{\HH}(s;z)\UU^{(n)}(s;z).
\end{gather*}
The first one gives two new relations
\begin{gather*}
\aa_n\bbeta_n(\BB_0-\bm L_{n-1}^*)+(\BB_0-\bm L_{n}^*)\aa_n^{-1}\bb_n(s\Id_N-\bb_n)=\big[\BB^2+\alpha\BB+\big[\BB_0,\bm a_{n,n-1}\big],\bb_n\big],
\end{gather*}
and
\begin{gather*}
(s\Id_N-\bb_n)(\BB_0-\bm L_{n}^*) -(\BB_0-\bm L_{n}^*)\aa_n^{-1}\bb_n\aa_n=\big(\BB^2+\alpha\BB+ [\BB_0,\bm a_{n,n-1} ]\big)\aa_n\\
\qquad{} -\aa_n\big(\BB_n^2+\alpha\BB_n+\big[\bm L_n,\gg_n\bm a_{n,n-1}\gg_n^{-1}\big]\big)^*.
\end{gather*}
The second one gives another two new relations
\begin{gather*}
(\BB_0-\bm L_{n+1}^*)\bbeta_{n+1}-\bbeta_n(\BB_0-\bm L_{n-1}^*)=\big[\BB^2+\alpha\BB+ [\BB_0,\bm a_{n,n-1} ],\aalpha_n\big]+ [\aalpha_n,\BB_0\aalpha_n ],
\end{gather*}
and
\begin{gather*}
\BB^2+\alpha\BB+ [\BB_0,\bm a_{n,n-1} ]-\big(\BB_n^2+\alpha\BB_n+\big[\bm L_n,\gg_n\bm a_{n,n-1}\gg_n^{-1}\big]\big)^*=\BB_0\aalpha_n-\aalpha_n\bm L_n^*.
\end{gather*}
Also from Proposition \ref{propg} we can get two new relations
\begin{gather*}
\aa_n\bbeta_n(\BB_0-\bm L_{n-1}^*)-(s\Id_N-\bb_n)(\BB_0-\bm L_{n}^*)\aa_n^{-1}\bb_n+s\big(\hat{\BB}_n^2+\alpha\hat{\BB}_n\big)\\
=(s\Id_N-\bb_n)\big(\BB^2+\alpha\BB+ [\BB_0,\bm a_{n,n-1} ]\big)+\aa_n\big(\BB_n^2+\alpha\BB_n+\big[\bm L_n,\gg_n\bm a_{n,n-1}\gg_n^{-1}\big]\big)^*\aa_n^{-1}\bb_n,
\end{gather*}
and
\begin{gather*}
\aa_n^{-1}\bb_n\big(\BB^2+\alpha\BB+ [\BB_0,\bm a_{n,n-1} ]\big)-\big(\BB_n^2+\alpha\BB_n+\big[\bm L_n,\gg_n\bm a_{n,n-1}\gg_n^{-1}\big]\big)^*\aa_n^{-1}\bb_n\\
\qquad{} +\bbeta_n(\BB_0-\bm L_{n-1}^*)+\aa_n^{-1}\bb_n(\BB_0-\bm L_{n}^*)\aa_n^{-1}\bb_n=0.
\end{gather*}
Combining these last two ones we get the reduction
\begin{gather*}
\hat{\BB}_n^2+\alpha\hat{\BB}_n=\BB^2+\alpha\BB+ [\BB_0,\bm a_{n,n-1} ]+(\BB_0-\bm L_{n}^*)\aa_n^{-1}\bb_n.
\end{gather*}

\subsection*{Acknowledgements}

M.C.~acknowledges the financial support of the Universidad Nacional Aut\'onoma de M\'exico (UNAM) and the Unit\'e Mixte International (UMI) ``Laboratoire Solomon Lefschetz'', and thanks their staff for the hospitality during his stay in Mexico. We both acknowledge the financial support of the Instituto de Ciencias Matem\'aticas (ICMAT) for our stay in Madrid during the thematic program ``Orthogonal polynomials and special functions in Mathematical Physics and Approximation Theory'', and we are particularly grateful to David G\'omez-Ullate for his invitation to participate. Finally, the work of the first author is also supported by the project IPaDEGAN (H2020-MSCA-RISE-2017), grant number 778010 (European Union), and the work of the second one by PAPIIT-DGAPA-UNAM grant IA102617 (Mexico).

\pdfbookmark[1]{References}{ref}
\LastPageEnding

\end{document}